\documentclass[aip,jmp,10pt]{revtex4-1}
\usepackage{amssymb}

\usepackage{amsmath,mathrsfs}
\usepackage{graphicx,color}
\usepackage{epstopdf}
\usepackage{mathrsfs}
\usepackage{epsfig}

\newtheorem{theorem}{Theorem}

\newtheorem{lemma}[theorem]{Lemma}

\newtheorem{proposition}[theorem]{Proposition}

\newenvironment{proof}[1][Proof]{\textbf{#1.} }{\ \rule{0.5em}{0.5em}}

\begin{document}

\title{The Stratonovich Formulation of Quantum Feedback Network Rules}
\author{John Gough}
\email{jug@aber.ac.uk}
\date{\today}

\affiliation{Department of Physics, Aberystwyth University, SY23 3BZ, Wales,
United Kingdom}

\begin{abstract}
We express the rules for forming quantum feedback networks using the Stratonovich form of quantum stochastic calculus rather than the It\={o}, or $SLH$ form. Remarkably the feedback reduction rule implies that we obtain the Schur complement of the matrix of Stratonovich coupling operators where we short out the internal input/output coefficients. 
\end{abstract}

\maketitle

\section{Introduction}
Quantum stochastic differential equations \cite{HP,Gardiner} give the mathematical framework for modelling open quantum systems where
one wish to explicitly take account of inputs and outputs, sometimes referred to as the SLH framework. There are explicit rules for forming a
quantum feedback network where various outputs from component systems are fed back in as inputs elsewhere in the network \cite{GJ_QFN,GJ_Series}. The techniques have since been widely applied to model quantum optical systems both theoretically and experimentally \cite{I_12}-\cite{ZJLYXP_2015}. The following paper gives the formulation of the rules in terms of the Stratonovich form of the quantum stochastic calculus
\cite{Gardiner,Chebotarev,G_Wong-Zakai}. The Stratonovich form is closer to a Hamiltonian description, and has previously shown revealed
some interesting representations of open quantum systems \cite{G_JMP_2015,G_OSID_2015}.
The reformulation the quantum feedback connection rules for the Stratonovich representation will of course follow as an exercise in converting from the existing It\={o} rules, however, the rather surprising result is that the resulting rule that emerges is the direct Schur complementation procedure. This is somewhat surprising as there is no obvious simple algebraic elimination of internal inputs and outputs
going on at the level of the Hamiltonian. Also, the simplest network consisting of two systems in series is not encouraging when analyzed in terms of the Stratonovich form. Nevertheless the rule that emerges is remarkable direct and succinct. It is possible that the Stratonovich form of the feedback rule may have some numerical advantages in modelling quantum feedback networks, for instance using computational packages such as QHDL\cite{QHDL}.

\subsection{SLH Framework}
In the Markovian model of an open quantum system we consider a fixed Hilbert
space $\mathfrak{h}_{0}$ for the system and a collection of independent
quantum white noises $b_{k}\left( t\right) $ labeled by $k$ belonging to
some discrete set $\mathsf{k}=\left\{ 1,\cdots ,n\right\} $, that is, we
have 
\begin{eqnarray*}
\left[ b_{j}\left( t\right) ,b_{k}\left( s\right) ^{\ast }\right] =\delta
_{jk}\delta \left( t-s\right) .
\end{eqnarray*}
The Schr\"{o}dinger equation is 
\begin{eqnarray}
\dot{U}\left( t\right) =-i\Upsilon \left( t\right) \,U\left( t\right) 
\label{eq:Schrodinger}
\end{eqnarray}
where the stochastic Hamiltonian takes the form 
\begin{eqnarray}
\Upsilon \left( t\right) =E_{00}+\sum_{j\in \mathsf{k}}E_{j0}b_{j}\left(
t\right) ^{\ast }+\sum_{k\in \mathsf{k}}E_{0k}\ b_{k}\left( t\right)
+\sum_{j,k\in \mathsf{k}}E_{jk}b_{j}\left( t\right) ^{\ast }b_{k}\left(
t\right) .  \label{eq:Upsilon}
\end{eqnarray}
Here we assume that the $E_{\alpha \beta }$ are operators on $\mathfrak{h}%
_{0}$ with $E_{\alpha \beta }^{\ast }=E_{\beta \alpha }$. (In the course of
this paper we will take them to be bounded so as to avoid technical
distractions.) We may write 
\begin{eqnarray*}
\Upsilon \left( t\right) =\left. 
\begin{array}{c}
\left[ 1,b_{\mathsf{k}}\left( t\right) ^{\ast }\right]  \\ 
\quad 
\end{array}
\right. \mathbf{E}\left[ 
\begin{array}{c}
1 \\ 
b_{\mathsf{k}}\left( t\right) 
\end{array}
\right] 
\end{eqnarray*}
where 
\begin{eqnarray*}
b_{\mathsf{k}}\left( t\right) =\left[ 
\begin{array}{c}
b_{1}\left( t\right)  \\ 
\vdots  \\ 
b_{n}\left( t\right) 
\end{array}
\right] ,\quad \mathbf{E}=\left[ 
\begin{array}{cc}
E_{00} & E_{0\mathsf{k}} \\ 
E_{\mathsf{k}0} & E_{\mathsf{kk}}
\end{array}
\right] 
\end{eqnarray*}
with $E_{0\mathsf{k}}=\left[ E_{01},\cdots ,E_{0n}\right] $, $E_{\mathsf{k}%
0}=E_{0\mathsf{k}}^{\ast }$ and $E_{\mathsf{kk}}$ the $n\times n$ matrix
with entries $E_{jk}$ with $j,k\in \mathsf{k}$. The Schr\"{o}dinger equation
(\ref{eq:Schrodinger}) is interpreted as the Stratonovich quantum stochastic
differential equation 
\begin{eqnarray}
dU\left( t\right)  &=&-i\bigg\{E_{00}\otimes dt+\sum_{j\in \mathsf{k}%
}E_{j0}\otimes dB_{j}\left( t\right) ^{\ast } \nonumber \\
&&+\sum_{k\in \mathsf{k}}E_{0k}\otimes dB_{k}\left( t\right) +\sum_{j,k\in 
\mathsf{k}}E_{jk}\otimes d\Lambda _{jk}\left( t\right) \bigg\}\circ U\left(
t\right) ,
\label{eq:S_QSDE}
\end{eqnarray}
which may be readily converted into the quantum It\={o} form of Hudson and
Parthasarathy. (In fact the latter is accomplished by Wick ordering the
noise fields $b_{k}(t)$ and $b_{k}^{\ast }(t)$ in (\ref{eq:Schrodinger}).)

In the theory of quantum feedback networks, we consider interconnected open
Markovian models in the limit where the connections have zero time delay. If
we have several components, separately described by stochastic Hamiltonians $%
\Upsilon ^{\left( i\right) }\left( t\right) $ for $i=1,\cdots ,m$, then the
natural Hamiltonian describing all components would be
\begin{eqnarray}
\Upsilon \left( t\right)  =\sum_{i=1}^{m}\Upsilon ^{\left( i\right)
}\left( t\right)  
=\sum_{i=1}^{m}\left. 
\begin{array}{c}
\left[ 1,b_{\mathsf{k}\left( i\right) }\left( t\right) ^{\ast }\right]  \\ 
\quad 
\end{array}
\right. \mathbf{E}^{\left( i\right) }\left[ 
\begin{array}{c}
1 \\ 
b_{\mathsf{k}\left( i\right) }\left( t\right) 
\end{array}
\right] 
\label{eq:S_parallel}
\end{eqnarray}
where $\mathsf{k}\left( i\right) $ is the set of input labels for the $i$th
component, and the coupling terms are by $\mathbf{E}^{\left( i\right)
}=\left[ 
\begin{array}{cc}
E_{00}^{\left( i\right) } & E_{0\mathsf{k}\left( i\right) }^{\left( i\right)
} \\ 
E_{\mathsf{k}\left( i\right) 0}^{\left( i\right) } & E_{\mathsf{k}\left(
i\right) \mathsf{k}\left( i\right) }^{\left( i\right) }
\end{array}
\right] $. (At this stage we make no assumptions about whether the coupling
operators of different components commute or not: mathematically we just
take all $E_{\alpha \beta }^{\left( i\right) }$ to be defined one the same
Hilbert space $\mathfrak{h}_{0}$ describing the all the components
collectively and work at this level of generality.) The total Hamiltonian
is then
\begin{eqnarray*}
\Upsilon ^{\mathrm{parallel}}\left( t\right) =\left. 
\begin{array}{c}
\left[ 1,b_{\mathsf{k}}\left( t\right) ^{\ast }\right]  \\ 
\quad 
\end{array}
\right. \mathbf{E}^{\mathrm{parallel}}\left[ 
\begin{array}{c}
1 \\ 
b_{\mathsf{k}}\left( t\right) 
\end{array}
\right] 
\end{eqnarray*}
where $\mathsf{k}=\mathsf{k}\left( 1\right) \cup \cdots \cup \mathsf{k}\left( m\right) $ is the collection of labels for all components, and 
\begin{eqnarray}
\mathbf{E}^{\mathrm{parallel}}=\left[ 
\begin{array}{cccc}
\sum_{i=1}^{m}E_{00}^{\left( i\right) } & E_{0\mathsf{k}\left( 1\right)
}^{\left( 1\right) } & \cdots  & E_{0\mathsf{k}\left( m\right) }^{\left(
m\right) } \\ 
E_{\mathsf{k}\left( 1\right) 0}^{\left( 1\right) } & E_{\mathsf{k}\left(
1\right) \mathsf{k}\left( 1\right) }^{\left( 1\right) } &  & 0 \\ 
\vdots  &  & \ddots  &  \\ 
E_{\mathsf{k}\left( m\right) 0}^{\left( m\right) } & 0 &  & E_{\mathsf{k}%
\left( m\right) \mathsf{k}\left( m\right) }^{\left( m\right) }
\end{array}
\right] .  
\label{eq:E_parallel}
\end{eqnarray}

This gives us our first network rule - how to assemble the $m$ components
together into a single model before connections are made. The second rule
deals with making connections. Here we must divide our set of fields into
two groups - those that are external and those that are to be fed back in as
internal drives. We denote the label sets for these as $\mathsf{e}$ and $%
\mathsf{i}$ respectively so that the total set $\mathsf{k}$ is the disjoint
union $\mathsf{e}\cup \mathsf{i}$. The Hamiltonian $\Upsilon ^{\mathrm{parallel%
}}\left( t\right) $ may then be decomposed as
\begin{eqnarray*}
\Upsilon ^{\mathrm{parallel}}\left( t\right) =\left. 
\begin{array}{c}
\left[ 1,b_{\mathsf{e}}\left( t\right) ^{\ast },b_{\mathsf{i}}\left(
t\right) ^{\ast }\right]  \\ 
\quad 
\end{array}
\right. \left[ 
\begin{array}{ccc}
E_{00} & E_{0\mathsf{e}} & E_{0\mathsf{i}} \\ 
E_{\mathsf{e}0} & E_{\mathsf{ee}} & E_{\mathsf{ei}} \\ 
E_{\mathsf{i}0} & E_{\mathsf{ie}} & E_{\mathsf{ii}}
\end{array}
\right] \left[ 
\begin{array}{c}
1 \\ 
b_{\mathsf{e}}\left( t\right)  \\ 
b_{\mathsf{i}}\left( t\right) 
\end{array}
\right] .
\end{eqnarray*}
Applying the feedback connections we should obtain a reduced model where the
internal inputs have been eliminated leaving only the set $\mathsf{e}$ of
external fields, that is, we should obtain a stochastic Hamiltonian of the
form
\begin{eqnarray*}
\Upsilon ^{\mathrm{fb}}\left( t\right) =\left. 
\begin{array}{c}
\left[ 1,b_{\mathsf{e}}\left( t\right) ^{\ast }\right]  \\ 
\quad 
\end{array}
\right. \mathbf{E}^{\mathrm{fb}}\left[ 
\begin{array}{c}
1 \\ 
b_{\mathsf{e}}\left( t\right) 
\end{array}
\right] .
\end{eqnarray*}
The expression for the reduced coefficients has been derived for the It\={o}
form, however, the remarkable result presented here is that the feedback
reduction formula for the Stratonovich form is actually the Schur complement
of the matrix $\mathbf{E}^{\mathrm{parallel}}$ describing the open loop
network where we short out the internal blocks. Under the assumption that the matrix $E_{\mathsf{ii}}$ of operators is invertible, we shall show that 
\begin{eqnarray}
\mathbf{E}^{\mathrm{fb}}\equiv \left[ 
\begin{array}{cc}
E_{00}^{\mathrm{fb}} & E_{0\mathsf{e}}^{\mathrm{fb}} \\ 
E_{\mathsf{e}0}^{\mathrm{fb}} & E_{\mathsf{ee}}^{\mathrm{fb}}
\end{array}
\right] =\left[ 
\begin{array}{cc}
E_{00} & E_{0\mathsf{e}} \\ 
E_{\mathsf{e}0} & E_{\mathsf{ee}}
\end{array}
\right] -\left[ 
\begin{array}{c}
E_{0\mathsf{i}} \\ 
E_{\mathsf{ei}}
\end{array}
\right] E_{\mathsf{ii}}^{-1}\left[ 
\begin{array}{cc}
E_{\mathsf{i}0} & E_{\mathsf{ie}}
\end{array}
\right] .
\label{eq:E_fb}
\end{eqnarray}

\subsection{Notation}

Let $\mathfrak{h}$ be a fixed Hilbert space. Given a countable set $\mathsf{j%
}$ of labels we set $\mathbb{C}^{\mathsf{j}}$ which is then a Hilbert space
spanned by a collection of orthonormal vectors $\left\{ e_{j}:j\in \mathsf{j}%
\right\} $, we may take as canonical basis. The Hilbert space $\mathfrak{h}%
\otimes \mathbb{C}^{\mathsf{j}}$ may then be represented as $\oplus _{j\in 
\mathsf{j}}\mathfrak{h}$, that is, as the set of vectors $\Psi =\left[ \psi
_{j}\right] _{j\in \mathsf{j}}$ with each $\psi _{j}\in \mathfrak{h}$ and $%
\sum_{j\in \mathsf{j}}\left\| \psi _{j}\right\| ^{2}<\infty $. An operator $X
$ on $\mathfrak{h}\otimes \mathbb{C}^{\mathsf{j}}$ may likewise we
represented as the array $\left[ X_{jk}\right] _{j,k\in \mathsf{j}}$ where
each $X_{jk}$ is an operator on $\mathfrak{h}$, so that $X\Psi =X\left[ \psi
_{j}\right] _{j\in \mathsf{j}}=\left[ \sum_{k\in \mathsf{j}}X_{jk}\psi _{k}%
\right] _{j\in \mathsf{j}}$.

Let $\mathsf{a}$ be a subset of $\mathsf{j}$. We shall understand that a block matrix $X_{\mathsf{aa}}$ is invertible
(with inverse $Y_{\mathsf{aa}}$) to mean the obvious property that the
system of equations
\begin{equation*}
\sum_{a^{\prime }\in \mathsf{a}}X_{aa^{\prime }}\phi _{a^{\prime }}=\psi
_{a},\qquad \left( \forall \quad a\in \mathsf{a}\right) ,
\end{equation*}
has an unique solution $\left( \phi _{a}\right) _{a\in \mathsf{a}}$ for any
given $\left( \psi _{a}\right) _{a\in \mathsf{a}}$ (given by $\phi
_{a}=\sum_{a^{\prime }\in \mathsf{a}}Y_{aa^{\prime }}\psi _{a^{\prime }}$). 
In particular, we shall write $I_{\mathsf{a}}$ for the identity on $\mathfrak{h}\otimes \mathbb{C}^{\mathsf{a}}$.

 In general if $\mathsf{a}$ and $\mathsf{b}$ are
subsets of the label set $\mathsf{j}$ we may set $X_{\mathsf{ab}}=\left[
X_{ab}\right] _{a\in \mathsf{a},b\in \mathsf{b}}$ which defined as sub-block
operator. If $\mathsf{a}$ and $\mathsf{b}$ are nonempty disjoint subsets
such that $\mathsf{j}=\mathsf{a}\cup \mathsf{b}$ then we may reorder the
operator into sub-blocks as $X\equiv \left[ 
\begin{array}{cc}
X_{\mathsf{aa}} & X_{\mathsf{ab}} \\ 
X_{\mathsf{ba}} & X_{\mathsf{bb}}
\end{array}
\right] $ corresponding to the direct sum decomposition $\mathbb{C}^{\mathsf{%
j}}\cong \mathbb{C}^{\mathsf{a}}\oplus \mathbb{C}^{\mathsf{b}}$. In such
cases we define the Schur complement of $X$ to be 
\begin{eqnarray}
\underset{\mathsf{b}}{\mathrm{Schur}}X\triangleq X_{\mathsf{aa}}-X_{\mathsf{ab}%
}X_{\mathsf{bb}}^{-1}X_{\mathsf{ba}}
\end{eqnarray}
where we shall always assume that $X_{\mathsf{bb}}$ is invertible as an
operator on $\mathfrak{h}\otimes \mathbb{C}^{\mathsf{b}}$. Specifically, we
say that this is the Schur complement of $X$ obtained by shortening the set
of indices $\mathsf{b}$.

A key property that we shall use is that the order in which successive
shortening of indices are applied is not important. In particular, 
\begin{eqnarray*}
\underset{\mathsf{b}_{1}\cup \cdots \cup \mathsf{b}_{n}}{\mathrm{Schur}} \, X=%
\underset{\mathsf{b}_{1}}{\mathrm{Schur}}\cdots \underset{\mathsf{b}_{n}}{%
\mathrm{Schur}}X
\end{eqnarray*}
for any disjoint sets $\mathsf{b}_{1},\cdots ,\mathsf{b}_{n}\subset \mathsf{j%
}$.

\section{Quantum Feedback Networks}

\bigskip

\subsection{Quantum Stochastic Evolutions}

We recall the Hudson-Parthasarathy theory of quantum stochastic evolutions\cite{HP,partha}
on Hilbert spaces of the form 
\begin{eqnarray*}
\mathfrak{H}=\mathfrak{h}_{0}\otimes \Gamma \left( \mathfrak{K}\otimes
L^{2}[0,\infty )\right)
\end{eqnarray*}
where $\mathfrak{h}_{0}$ is a fixed Hilbert space, called the \emph{initial
space}, and $\mathfrak{K}$ is a fixed Hilbert space called the \emph{%
internal space}. Specifically, $\mathfrak{K}$ is the multiplicity space
(also known as the color space) of the noise and takes the form 
\begin{eqnarray*}
\mathfrak{K}=\mathbb{C}^{\mathsf{k}}
\end{eqnarray*}
where $\mathfrak{K}=\left\{ 1,\cdots ,n\right\} $ are the labels of the
input noise fields driving the open quantum system with Hilbert space $%
\mathfrak{h}_{0} $. As before, let $\left\{ e_{k}:k\in \mathsf{k}\right\} $
be the canonical orthonormal basis for $\mathfrak{K}$. The annihilation
processes are defined, for each $k\in \mathsf{k}$, by 
\begin{eqnarray*}
B_{k}\left( t\right) \triangleq a\left( e_{k}\otimes 1_{\left[ 0,t\right]
}\right)
\end{eqnarray*}
where $a\left( \cdot \right) $ is the annihilation functor from $\mathfrak{K}%
\otimes L^{2}[0,\infty )$ to the Fock space $\Gamma \left( \mathfrak{K}%
\otimes L^{2}[0,\infty )\right) $ and $1_{\left[ 0,t\right] }$ is the
indicator function for the interval $\left[ 0,t\right] $. Its adjoint $%
B_{k}\left( t\right) ^{\ast }$ is the creation process, and scattering is
described by the processes 
\begin{eqnarray*}
\Lambda _{jk}\left( t\right) \triangleq d\Gamma \left( |e_{j}\rangle \langle
e_k|\otimes \pi _{\left[ 0,t\right] }\right)
\end{eqnarray*}
where $d\Gamma \left( \cdot \right) $ is the differential second
quantization functor and $\pi _{\left[ 0,t\right] }$ is the operator of
pointwise multiplication by $1_{\left[ 0,t\right] }$ on $L^{2}[0,\infty )$.

As is well known $\mathfrak{H}$ decomposes as $\mathfrak{H}_{\left[ 0,t%
\right] }\otimes \mathfrak{H}_{\left( t,\infty \right) }$ for each $t>0$
where $\mathfrak{H}_{\left[ 0,t\right] }=\mathfrak{h}_{0}\otimes \Gamma
\left( \mathfrak{k}\otimes L^{2}[0,t)\right) $ and $\mathfrak{H}_{\left(
t,\infty \right) }=\Gamma \left( \mathfrak{k}\otimes L^{2}(t,\infty )\right) 
$. We shall write $\mathfrak{A}_{t]}$ for the space of operators on $%
\mathfrak{H}$ that act trivially on the future component $\mathfrak{H}%
_{\left( t,\infty \right) }$. A quantum stochastic process $X_{t}=\left\{
X_{t}:t\geq 0\right\} $ is said to be \emph{adapted} if $X_{t}\in %
\mathfrak{A}_{t]}$ for each $t\geq 0$.

Taking $\left\{ x_{\alpha \beta }\left( t\right) :t\geq 0\right\} $ to be a
family of adapted quantum stochastic processes, their quantum stochastic
integral is $X_{t}=\int_{0}^{t}x_{\alpha \beta }\left( s\right) dB^{\alpha
\beta }\left( t\right) $ which is shorthand for 
\begin{eqnarray*}
\int_{0}^{t}x_{00}\left( s\right) ds+\sum_{j\in \mathsf{j}}\int_{0}^{t} x_{j0}\left(
s\right) dB_{j}\left( s\right) ^{\ast } 
+\sum_{k\in \mathsf{j}} \int_{0}^{t} x_{0k}\left(
s\right) dB_{k}\left( s\right) +\sum_{j,k\in \mathsf{j}} \int_{0}^{t}
x_{jk}\left( s\right) d\Lambda _{jk}\left( s\right)
\end{eqnarray*}
and where the differentials are understood in the It\={o} sense. Given a
similar quantum It\={o} integral $Y_{t}$, with $dY_{t}=y_{\alpha \beta
}\left( t\right) dB^{\alpha \beta }\left( t\right) $, we have the quantum
It\={o} product rule 
\begin{eqnarray}
d\left( X_{t}.Y_{t}\right) =dX_{t}.Y_{t}+X_{t}.dY_{t}+dX_{t}.dY_{t},
\label{Ito formula}
\end{eqnarray}
with the It\={o} correction given by 
\begin{eqnarray}
dX_{t}.dY_{t}=x_{\alpha k}\left( t\right) y_{k\beta }\left( t\right)
\,dB^{\alpha \beta }\left( t\right) .  \label{Ito correction}
\end{eqnarray}

The coefficients $\left\{ x_{\alpha \beta }\left( t\right) \right\} $ may be
assembled into a matrix 
\begin{eqnarray}
\mathbf{X}\left( t\right) =\left[ 
\begin{tabular}{ll}
$x_{00}\left( t\right) $ & $x_{0\mathsf{k}}\left( t\right) $ \\ 
$x_{\mathsf{k}0}\left( t\right) $ & $x_{\mathsf{kk}}\left( t\right) $%
\end{tabular}
\right] \in \mathfrak{A}_{t]}^{\left( 1+n\right) \times \left( 1+n\right) },
\label{Ito matrix}
\end{eqnarray}
which we term the \emph{It\={o} matrix} for the process. Note that in terms
of our earlier conventions, the appropriate index set is the $0\cup \mathsf{k%
}$ which has $1+n$ elements.

(Here we use the convention that $x_{0\mathsf{k}}\left( t\right) $\ denotes
the row vector with entries $\left( x_{0j}\left( t\right) \right) _{j=1}^{n}$%
, etc. The It\={o} matrix for a product $X_{t}Y_{t}$ of quantum It\={o}
integrals will then be given by $\mathbf{X}\left( t\right) Y\left( t\right)
+X\left( t\right) \mathbf{Y}\left( t\right) +\mathbf{X}\left( t\right) 
\mathbf{\hat{\delta}Y}\left( t\right) $, where $\mathbf{\hat{\delta}}%
\triangleq \left[ 
\begin{tabular}{l|l}
$0$ & $0$ \\ \hline
$0$ & $I_{\mathsf{k}}$%
\end{tabular}
\right] $.

The general form of the constant operator-coefficient quantum stochastic
differential equation for an adapted unitary process $U$ is 
\begin{eqnarray}
dU\left( t\right) = \bigg\{ -\left( \frac{1}{2}L_{\mathsf{k}}^{\ast }L_{%
\mathsf{k}}+iH\right) dt+\sum_{j\in \mathsf{k}}L_{j}dB_{j}\left( t\right)
^{\ast }  -\sum_{j,k\in \mathsf{k}}S_{jk}L_{k}dB_{k}\left( t\right) +\sum_{j,k\in 
\mathsf{k}}(S_{jk}-\delta _{jk})d\Lambda _{jk}\left( t\right) \bigg\} %
U\left( t\right)
\label{eq:Ito_QSDE}
\end{eqnarray}
where the $S_{jk},L_{j}$ and $H$ are operators on the initial Hilbert space
with $S_{\mathsf{kk}}=\left[ S_{jk}\right] _{j,k\in \mathsf{k}}$ is unitary
and $H$ self-adjoint. (We use the convention that $L_{\mathsf{k}}=\left[
L_{k}\right] _{k\in \mathsf{k}}$ and that \ $L_{\mathsf{k}}^{\ast }L_{%
\mathsf{k}}=\sum_{k\in \mathsf{k}}L_{k}^{\ast }L_{k}$.) The corresponding
It\={o} matrix of coefficients is 
\begin{eqnarray}
\mathbf{G}=\left[ 
\begin{array}{cc}
-\left( \frac{1}{2}L_{\mathsf{k}}^{\ast }L_{\mathsf{k}}+iH\right) & -S_{%
\mathsf{kk}}L_{\mathsf{k}}^{\ast } \\ 
L_{\mathsf{k}} & S_{\mathsf{kk}}-I_{\mathsf{k}}
\end{array}
\right] ,
\end{eqnarray}
which we have previously called the It\={o} generator matrix. The triple $%
\left( S,L,H\right) $ are termed the Hudson-Parthasarathy parameters of the
open system evolution. Explicitly, we have that\cite{G_Wong-Zakai}
\begin{eqnarray*}
S_{\mathsf{kk}} =\frac{I_{\mathsf{k}}-\frac{i}{2}E_{\mathsf{kk}}}{I_{%
\mathsf{k}}+\frac{i}{2}E_{\mathsf{kk}}},\quad L_{\mathsf{k}}=-i\left( I_{%
\mathsf{k}}+\frac{i}{2}E_{\mathsf{kk}}\right) ^{-1}E_{\mathsf{k}0}, \quad
H =E_{00}+\frac{1}{2}\mathrm{Im}\left\{ E_{0\mathsf{k}}\left( I_{\mathsf{k}}+%
\frac{i}{2}E_{\mathsf{kk}}\right) ^{-1}E_{\mathsf{k}0}\right\}
\end{eqnarray*}
where Im$X$ means $\frac{1}{2i}\left( X-X^{\ast }\right) $. Note that we may invert to get
\begin{eqnarray}
E_{\mathsf{kk}} = \frac{2}{i} \frac{I_{\mathsf{k}}-S_{\mathsf{kk}}}{I_{%
\mathsf{k}}+S_{\mathsf{kk}}},\label{eq:E_from_S}
\end{eqnarray}
provided that $I_{\mathsf{k}}+S_{\mathsf{kk}}$ is invertible. We note that there will exist SLH models that do not
possess Stratonovich representations. A simple example is an optical mirror for which $S \equiv -1$.

\begin{figure}[htbp]
\centering
\includegraphics[width=0.40\textwidth]{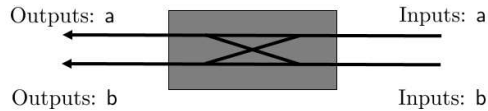}
\caption{A single SLH component with input labels $\mathsf{k}$ split into
two groups $\mathsf{a}$ and $\mathsf{b}$.}
\label{fig:SS_single_component}
\end{figure}

In Figure \ref{fig:SS_single_component} we have partitions the inputs and
outputs into two groups. Here we have the block partition 
\begin{eqnarray*}
S_{\mathsf{kk}}\equiv \left[ 
\begin{array}{cc}
S_{\mathsf{aa}} & S_{\mathsf{ab}} \\ 
S_{\mathsf{ba}} & S_{\mathsf{bb}}
\end{array}
\right] ,\quad L_{\mathsf{k}}\equiv \left[ 
\begin{array}{c}
L_{\mathsf{a}} \\ 
L_{\mathsf{b}}
\end{array}
\right] .
\end{eqnarray*}
Note that we did not need to group the outputs in the same way as the inputs.

\subsubsection{Belavkin-Holevo Matrix Representation}
As is well-known, the Heisenberg group gives a matrix representation of the Lie algebra of the usual canonical commutation relations,
specifically in terms of upper triangular matrices. Independently, Belavkin \cite{Belavkin_1988} and Holevo
\cite{Holevo_1989} developed the analogous representation in the setting for quantum stochastic calculus.
We consider the mapping from It\={o} matrices $\mathbf{X}\in \mathfrak{A}%
^{\left( 1+n\right) \times \left( 1+n\right) }$ to associated \textit{%
Belavkin-Holevo matrices} 
\begin{eqnarray}
\mathbb{X}=\mathscr{H}\left( \left[ 
\begin{tabular}{ll}
$x_{00}$ & $x_{0\mathsf{k}}$ \\ 
$x_{\mathsf{k}0}$ & $x_{\mathsf{kk}}$%
\end{tabular}
\right] \right) =\left[ 
\begin{array}{ccc}
0 & x_{0\mathsf{k}} & x_{00} \\ 
0 & x_{\mathsf{kk}} & x_{\mathsf{k}0} \\ 
0 & 0 & 0
\end{array}
\right] \in \mathfrak{A}^{\left( 1+n+1\right) \times \left( 1+n+1\right) },
\label{Belavkin matrix}
\end{eqnarray}
understood for each fixed time $t$. The Belavkin-Holevo matrices are $1+n+1$ square
dimensional and we shall write their label set as 
\begin{eqnarray*}
\mathsf{j}=\overline{0}\cup \mathsf{k}\cup \underline{0},
\end{eqnarray*}
that is, $\overline{0}$ labels the top row/column and $\underline{0}$ labels
the bottom row/column. We also introduce 
\begin{eqnarray*}
\begin{array}{cc}
\mathbb{I}\triangleq \left[ 
\begin{array}{ccc}
1 & 0 & 0 \\ 
0 & I_{\mathsf{k}} & 0 \\ 
0 & 0 & 1
\end{array}
\right] , & \mathbb{J}\triangleq \left[ 
\begin{array}{ccc}
0 & 0 & 1 \\ 
0 & I_{\mathsf{k}} & 0 \\ 
1 & 0 & 0
\end{array}
\right] .
\end{array}
\end{eqnarray*}
The \textit{twisted involution} on the set of Belavkin-Holevo matrices is defined
by 
\begin{eqnarray*}
\mathbb{X}^{\star }\triangleq \mathbb{JX}^{\dag }\mathbb{J}
\end{eqnarray*}
We have the following properties: 
\begin{eqnarray*}
\mathscr{H}\left( \mathbf{X}^{\dag }\right) =\mathscr{H}\left( \mathbf{X}%
\right) ^{\star }\mathrm{, and } \, \mathscr{H}\left( \mathbf{XPY}\right) =\mathscr{H}%
\left( \mathbf{X}\right) \mathscr{H}\left( \mathbf{Y}\right) \mathrm{.}
\end{eqnarray*}
The main advantage of using this representation is that the It\={o}
correction $\mathbf{XPY}$ can now be given as just the ordinary product $%
\mathbb{XY}$\ of the Belavkin-Holevo matrices.

The Belavkin-Holevo matrix associated with the It\={o} generating matrix is then 
\begin{eqnarray*}
\mathbb{G}=\mathscr{H}\left( \mathbf{G}\right) =\left[ 
\begin{array}{ccc}
0 & -S_{\mathsf{kk}}L_{\mathsf{k}}^{\ast } & -\left( \frac{1}{2}L_{\mathsf{k}%
}^{\ast }L_{\mathsf{k}}+iH\right) \\ 
0 & S_{\mathsf{kk}}-I_{\mathsf{k}} & L_{\mathsf{k}} \\ 
0 & 0 & 0
\end{array}
\right]
\end{eqnarray*}
The related matrix 
\begin{eqnarray}
\mathbb{V}=\mathbb{I}+\mathbb{G=}\left[ 
\begin{array}{ccc}
1 & -S_{\mathsf{kk}}L_{\mathsf{k}}^{\ast } & -\left( \frac{1}{2}L_{\mathsf{k}%
}^{\ast }L_{\mathsf{k}}+iH\right) \\ 
0 & S_{\mathsf{kk}} & L_{\mathsf{k}} \\ 
0 & 0 & 1
\end{array}
\right] 
\end{eqnarray}
is $\star $-unitary, that is 
\begin{eqnarray}
\mathbb{VV}^{\star }=\mathbb{I}=\mathbb{V}^{\star }\mathbb{V}.
\end{eqnarray}

\subsubsection{The Stratonovich Form}

One can given the Stratonovich form of the unitary dynamics. Here we define
the Stratonovich integral by the algebraic relation $dX\left( t\right) \circ
Y\left( t\right) \equiv dX\left( t\right) \,Y\left( t\right) +\frac{1}{2}%
dX\left( t\right) dY\left( t\right) $ however this agrees with the notion of
a mid-point rule \cite{Chebotarev}. For the unitary QSDE (\ref{eq:Ito_QSDE}) we have the
Stratonovich form (\ref{eq:S_QSDE}) which we may write as
\begin{eqnarray*}
dU\left( t\right) =-idE\left( t\right) \circ U\left( t\right)
\end{eqnarray*}
end where 
\begin{eqnarray*}
dE\left( t\right) =E_{00}dt+\sum_{j\in \mathsf{j}}E_{j0}dB_{j}\left(
t\right) ^{\ast }+\sum_{k\in \mathsf{j}}E_{0k}dB_{k}\left( t\right)
+\sum_{j,k\in \mathsf{j}}E_{jk}d\Lambda _{jk}\left( t\right) .
\end{eqnarray*}
Formally the stochastic Hamiltonian $\Upsilon (t)$ is the derivative of the quantum stochastic integral process $E(t)$.
As outlined above, we assemble these operators into the Stratonovich generator matrix,
\begin{eqnarray*}
\mathbf{E}=\left[ 
\begin{array}{cc}
E_{00} & E_{0\mathsf{k}} \\ 
E_{\mathsf{k}0} & E_{\mathsf{kk}}
\end{array}
\right]
\end{eqnarray*}
and we have the relationship 
\begin{eqnarray*}
\mathbf{G}=-i\mathbf{E}-\frac{i}{2}\mathbf{E\hat{\delta}G.}
\end{eqnarray*}
This is more clearly seen in terms of Belavkin-Holevo matrices. Set 
\begin{eqnarray*}
\mathbb{E}=\mathscr{H} \left( \mathbf{E}\right) =\left[ 
\begin{array}{ccc}
0 & E_{0\mathsf{k}} & E_{00} \\ 
0 & E_{\mathsf{kk}} & E_{\mathsf{k}0} \\ 
0 & 0 & 0
\end{array}
\right]
\end{eqnarray*}
then 
\begin{eqnarray*}
\mathbb{V}=\frac{\mathbb{I}-\frac{i}{2}\mathbb{E}}{\mathbb{I}+\frac{i}{2}%
\mathbb{E}}.
\end{eqnarray*}
In particular, we see that $\mathbb{V}$ is $\star $-unitary if and only if $%
\mathbb{E=E}^{\star }$ which in turn implies that $E_{\alpha \beta }^{\ast
}=E_{\beta \alpha }$ for all $\alpha ,\beta \in 0\cup \mathsf{k}$.

\subsection{Quantum Feedback Networks}

In \cite{GJ_Series} we introduced the series product describing the
situation where one SLH drives another, see Figure \ref{fig:SS_Series}. Here
the output of the first system $\left( S_{1},L_{1},H_{1}\right) $ is fed
forward as the input to the second system $\left( S_{2},L_{2},H_{2}\right) $
and the limit of zero time delay is assumed. (Note that the systems do not
technically have to be distinct and may have the same initial space!) In the
limit, the Hudson-Parthasarathy parameters of the composite system were
shown to be \cite{GJ_Series},\cite{GJ_QFN} 
\begin{eqnarray*}
S_{\mathrm{series}} &=&S_{2}S_{1}, \\
L_{\mathrm{series}} &=&L_{2}+S_{2}L_{1}, \\
H_{\mathrm{series}} &=&H_{1}+H_{2}+\mathrm{Im}\left\{ L_{2}^{\dag
}S_{2}L_{1}\right\} .
\label{eq:series_product}
\end{eqnarray*}

\begin{figure}[tbph]
\centering
\includegraphics[width=0.40\textwidth]{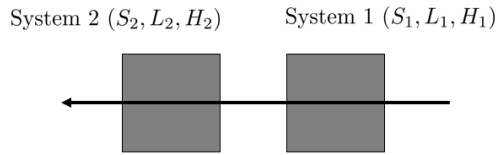}
\caption{Systems in series: Output of system 1 is the input of system 2.}
\label{fig:SS_Series}
\end{figure}
We refer to the associative group law 
\begin{eqnarray*}
\left( S_{2},L_{2},H_{2}\right) \vartriangleleft \left(
S_{1},L_{1},H_{1}\right) \triangleq \left(
S_{2}S_{1},L_{2}+S_{2}L_{1},H_{1}+H_{2}+\mathrm{Im} \{ L_{2}^{\dag
}S_{2}L_{1} \} \right)
\end{eqnarray*}
determined above as the series product. As remarked in \cite{GJ_QFN}, the
series product actually arises natural in Belavkin-Holevo matrix form as 
\begin{eqnarray*}
\mathbb{V}_{\mathrm{series}}=\mathbb{V}_{2}\mathbb{V}_{1}.
\end{eqnarray*}
The product is clearly associative, as one would expect physically, and the
general rule for several systems in series is then $\mathbb{V}_{\mathrm{series}%
}=\mathbb{V}_{n}\cdots \mathbb{V}_{2}\mathbb{V}_{1}$

\bigskip

More generally we gave the rules for construction an arbitrary quantum
feedback network where we have several open quantum components, each
described by a SLH model, and where various outputs of are fed back as
driving inputs with zero time delay, see Figure \ref{fig:SS_network}.

\begin{figure}[htbp]
\centering
\includegraphics[width=0.50\textwidth]{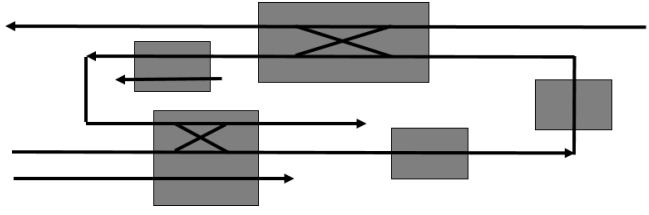}
\caption{A general quantum feedback network.}
\label{fig:SS_network}
\end{figure}

In fact, it suffices to give two rules for constructing arbitrary networks
of this type. The first step is to take the network description and break
all the connections leaving only the individual open loop description. We
can look upon this as a single SLH component, see Figure \ref
{fig:SS_parallel}.

\begin{figure}[htbp]
\centering
\includegraphics[width=0.20\textwidth]{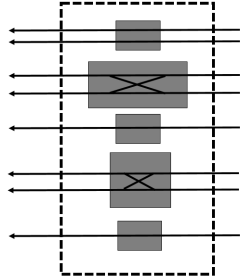}
\caption{The disconnected components in parallel.}
\label{fig:SS_parallel}
\end{figure}

Our first network rule is that the models $\left( S_{j},L_{j},H_{j}\right)
_{j=1}^{n}$ when concatenated in parallel, as sketched in Figure \ref{fig:SS_parallel},
correspond to the single SLH component 
\begin{eqnarray}
\boxplus _{j=1}^{n}\left( S_{j},L_{j},H_{j}\right) =\left( \left[ 
\begin{array}{ccc}
S_{1} & 0 & 0 \\ 
0 & \ddots & 0 \\ 
0 & 0 & S_{n}
\end{array}
\right] ,\left[ 
\begin{array}{c}
L_{1} \\ 
\vdots \\ 
L_{n}
\end{array}
\right] ,H_{1}+\cdots +H_{n}\right) .  \label{eq:SLH_parallel}
\end{eqnarray}

(Note that we have made no assumptions that the operators corresponding to
different components commute!)

The second network rule tells us how to connect the various internal inputs,
see Figure \ref{fig:SS_feedback}. 
\begin{figure}[tbph]
\centering
\includegraphics[width=0.50\textwidth]{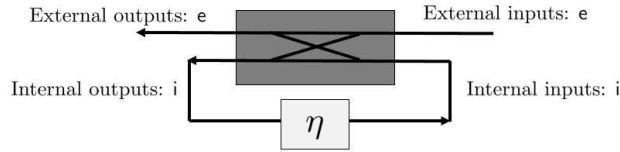}
\caption{Feedback}
\label{fig:SS_feedback}
\end{figure}
To this end we first divide the inputs and outputs into external and
internal groups. That is, the set $\mathsf{k}$ of all inputs is split into $%
\mathsf{e}$ which are the external inputs, and $\mathsf{i}$ which are the
internal. With respect to this decomposition, we have the block partition 
\begin{eqnarray*}
S_{\mathsf{kk}}=\left[ 
\begin{array}{cc}
S_{\mathsf{ee}} & S_{\mathsf{ei}} \\ 
S_{\mathsf{ie}} & S_{\mathsf{ii}}
\end{array}
\right] ,\,L_{\mathsf{k}}=\left[ 
\begin{array}{c}
L_{\mathsf{e}} \\ 
L_{\mathsf{i}}
\end{array}
\right] ,
\end{eqnarray*}
We may for convenience include an adjacency matrix $\eta $ which tells us
which internal input label corresponds to which internal output label. The
second rule then states that the reduced SLH model is 
\begin{eqnarray}
S^{\mathrm{fb}} &=&S_{\mathsf{ee}}+S_{\mathsf{ei}}\eta \left( I_{\mathsf{i}%
}-S_{\mathsf{ii}}\eta \right) ^{-1}S_{\mathsf{ie}},  \notag \\
L^{\mathrm{fb}} &=&L_{\mathsf{e}}+S_{\mathsf{ei}}\eta \left( I_{\mathsf{i}}-S_{%
\mathsf{ii}}\eta \right) ^{-1}L_{\mathsf{i}},  \notag \\
H^{\mathrm{fb}} &=&H+\sum_{i=\mathsf{i},\mathsf{e}}\mathrm{Im}L_{j}^{\dag }\eta
S_{j\mathsf{i}}\left( I_{\mathsf{i}}-S_{\mathsf{ii}}\eta \right) ^{-1}L_{%
\mathsf{i}}.  \label{eq:SLH_feedback}
\end{eqnarray}
We remark that the adjacency matrix is not essential and may be absorbed
into the matrix $S_{\mathsf{kk}}$. Indeed, we may think of the adjacency
matrix as a simple device performing a scattering only and the disconnected
model is $S_{\mathsf{kk}}\left[ 
\begin{array}{cc}
I_{\mathsf{e}} & 0 \\ 
0 & \eta
\end{array}
\right] =\left[ 
\begin{array}{cc}
S_{\mathsf{ee}} & S_{\mathsf{ei}}\eta \\ 
S_{\mathsf{ie}} & S_{\mathsf{ii}}\eta
\end{array}
\right] $ by the series product. As such all we have to do is to replace $S_{%
\mathsf{ei}}\eta $ by $S_{\mathsf{ei}}$ and $S_{\mathsf{ii}}\eta $ by $S_{%
\mathsf{ii}}$. In the following we shall always assume that this has already
done, so we work with (\ref{eq:SLH_feedback}) with $\eta =I_{\mathsf{i}}$.
We also note that we need the condition that the inverse of $I_{\mathsf{i}%
}-S_{\mathsf{ii}}$ exists: this is the condition needed to ensure that the
network is well-posed.

We follow Smolyanov and Truman\cite{ST} in expressing the feedback reduction rule in terms of the
Belavkin-Holevo matrices. This turns out to be most convenient for our purposes.
Let $\overline{0}\cup \mathsf{e}\cup \mathsf{i}\cup 
\underline{0}$ be the usual labeling of the entries of the matrix $\mathbb{V%
}$, then the feedback reduced matrix has components 
\begin{eqnarray*}
\left[ \mathbb{V}^{\mathrm{fb}}\right] _{\alpha \beta }=\mathbb{V}_{\alpha
\beta }+\mathbb{V}_{\alpha \mathsf{i}}\left( I_{\mathsf{i}}-\mathbb{V}_{%
\mathsf{ii}}\right) ^{-1}\mathbb{V}_{\mathsf{i}\beta }
\end{eqnarray*}
where $\alpha ,\beta $ now belong to $\overline{0}\cup \mathsf{e}\cup 
\underline{0}$. That is, the matrix $\mathbb{V}^{\mathrm{fb}}$ is a M\"{o}bius
transformation of the original $\mathbb{V}$. In fact, the feedback reduced $%
\mathbb{V}^{\mathrm{fb}}$ inherits the property of $\star $-unitarity from $%
\mathbb{V}$.

It is useful to reformulate this in terms of the associated $\mathbb{G}$
matrix. Indeed we have the following result 
\begin{eqnarray}
\mathbb{G}^{\mathrm{fb}}=\underset{\mathsf{i}}{\mathrm{Schur}}\mathrm{\thinspace }%
\mathbb{G.}  \label{eq:Schur_G}
\end{eqnarray}
Therefore, the feedback reduced $\mathbb{G}^{\mathrm{fb}}$ is just the
original $\mathbb{G}$ shortened by the internal inputs and outputs.

\section{The Stratonovich Version of Quantum Feedback Networks}

We now reformulate network theory in terms of the Stratonovich generating
matrices.

The first network rule is elementary. If we have $m$ components each with
input/output labels sets $\mathsf{k}\left( i\right) $ and Stratonovich
generating matrices $\mathbf{E}^{\left( i\right) }=\left[ 
\begin{array}{cc}
E_{00}^{\left( i\right) } & E_{0\mathsf{k}\left( i\right) }^{\left( i\right)
} \\ 
E_{\mathsf{k}\left( i\right) 0}^{\left( i\right) } & E_{\mathsf{k}\left(
i\right) \mathsf{k}\left( i\right) }^{\left( i\right) }
\end{array}
\right] $ for $i=1,\cdots ,m$, then the overall Stratonovich model for the
components in parallel is 
\begin{eqnarray*}
\boxplus _{i=1}^{m}\mathbf{E}^{\left( i\right) }=\mathbf{E}^{\mathrm{parallel}}
\end{eqnarray*}
which of course has the set on input labels $\mathsf{k}=\cup _{i=1}^{m}%
\mathsf{k}\left( i\right) $ and $\mathbf{E}^{\mathrm{parallel}}$ is given by (%
\ref{eq:E_parallel}).

The second rule involves splitting the labels up as external and internal: $%
\mathsf{k}=\mathsf{e}\cup \mathsf{i}$.

\begin{proposition}
Let $\mathbf{E}$ be the Stratonovich generator matrix labeled by $0\cup 
\mathsf{e}\cup \mathsf{i}$. The quantum feedback network obtained by feeding
the internal outputs back as internal inputs is well-posed if and only if the operator
\begin{eqnarray}
\mathscr{E}_{\mathsf{ii}} \triangleq
E_{\mathsf{ii}} -\frac{i}{2} E_{\mathsf{ie}} \left( I_{\mathsf{e}}+\frac{i}{2}
 E_{\mathsf{ee}} \right)^{-1} E_{\mathsf{ei}}
\label{eq:invertor}
\end{eqnarray}
 is strictly invertible.
\end{proposition}

\begin{proof}
We have that $\mathbb{V}_{\mathsf{ii}}=S_{\mathsf{ii}}$, and $S=\left[ 
\begin{array}{cc}
S_{\mathsf{ee}} & S_{\mathsf{ei}} \\ 
S_{\mathsf{ie}} & S_{\mathsf{ii}}
\end{array}
\right] =\left[ 
\begin{array}{cc}
I_{\mathsf{e}}+\frac{i}{2}E_{\mathsf{ee}} & \frac{i}{2}E_{\mathsf{ei}} \\ 
\frac{i}{2}E_{\mathsf{ie}} & I_{\mathsf{i}}+\frac{i}{2}E_{\mathsf{ii}}
\end{array}
\right] ^{-1}\left[ 
\begin{array}{cc}
I_{\mathsf{e}}-\frac{i}{2}E_{\mathsf{ee}} & \frac{i}{2}E_{\mathsf{ei}} \\ 
\frac{i}{2}E_{\mathsf{ie}} & I_{\mathsf{i}}-\frac{i}{2}E_{\mathsf{ii}}
\end{array}
\right] $ however a standard result in inverting block matrices shows that
\begin{eqnarray*}
\left[ 
\begin{array}{cc}
I_{\mathsf{e}}+\frac{i}{2}E_{\mathsf{ee}} & \frac{i}{2}E_{\mathsf{ei}} \\ 
\frac{i}{2}E_{\mathsf{ie}} & I_{\mathsf{i}}+\frac{i}{2}E_{\mathsf{ii}}
\end{array}
\right] ^{-1}\equiv \left[ 
\begin{array}{cc}
\left( I_{\mathsf{e}}+\frac{i}{2}E_{\mathsf{ee}}\right) ^{-1}-\frac{1}{4}%
\left( I_{\mathsf{e}}+\frac{i}{2}E_{\mathsf{ee}}\right) ^{-1}E_{\mathsf{ei}%
}\Delta _{\mathsf{ii}}E_{\mathsf{ie}} & \quad -\frac{i}{2}\left( I_{\mathsf{e}}+%
\frac{i}{2}E_{\mathsf{ee}}\right) ^{-1}E_{\mathsf{ei}}\Delta _{\mathsf{ii}}
\\ 
-\frac{i}{2}\Delta _{\mathsf{ii}}E_{\mathsf{ie}}\left( I_{\mathsf{e}}+\frac{i%
}{2}E_{\mathsf{ee}}\right) ^{-1} & \Delta _{\mathsf{ii}}
\end{array}
\right] ,
\end{eqnarray*}
where $\Delta _{\mathsf{ii}}=\left( 
I_{\mathsf{i}}+\frac{i}{2}E_{\mathsf{ii}}+\frac{1}{4}E_{\mathsf{ie}}
\left( I_{\mathsf{e}}+\frac{i}{2}E_{\mathsf{ee}}\right) ^{-1}
E_{\mathsf{ei}} \right) ^{-1} 
\equiv \left( I_{\mathsf{ i}}+\frac{i}{2}\mathscr{E}_{\mathsf{ii}} \right) ^{-1}$. Substituting in yields the
explicit form
\begin{equation*}
S_{\mathsf{ii}} \equiv \left( I_{\mathsf{i}}-\frac{i}{2}\mathscr{E}_{\mathsf{ii}} \right)
\left( I_{\mathsf{i}}+\frac{i}{2} \mathscr{E}_{\mathsf{ii}}\right) ^{-1}.
\end{equation*}

The well-posed property is that $I_{\mathsf{i}}-\mathbb{V}_{\mathsf{ii}}$ is
invertible, that is $I_{\mathsf{i}}-S_{\mathsf{ii}}$ is invertible. We see
that this is equivalent to the requirement that $\mathscr{E}_{\mathsf{ii}}$ be
invertible.
\end{proof}

\bigskip

We have seen that the condition for an SLH model to have a Stratonovich representation is that $I_{\mathsf{k}} + S_{\mathsf{kk}}$
is invertible. We additionally have the well-posed property that $I_{\mathsf{i}} - S_{\mathsf{ii}}$ 
is invertible

Before stating our main theorem, we have the following lemma which will be used in the proof.
\begin{lemma}
Let $\mathbb{G}= \mathscr{H} \left( \mathbf{G}\right) $ be the Belavkin-Holevo matrix
associated with a given It\={o} generating matrix $\mathbf{G}$. We double up
the usual labels $\overline{0}\cup \mathsf{k}\cup \underline{0}$ to get $%
\left\{ \overline{0}\cup \mathsf{k}\cup \underline{0}\right\} \cup \left\{ 
\overline{0}^{\prime }\cup \mathsf{k}^{\prime }\cup \underline{0}^{\prime
}\right\} $ where the $\mathsf{k}^{\prime }$ is a copy of $\mathsf{k}$. Then 
\begin{eqnarray*}
\mathbb{G}=-\underset{\overline{0}^{\prime }\cup \mathsf{k}^{\prime }\cup 
\underline{0}^{\prime }}{\mathrm{Schur}}\mathrm{\thinspace }\left[ 
\begin{array}{cc}
2\mathbb{I} & \sqrt{2}\mathbb{I} \\ 
\sqrt{2}\mathbb{I} & \mathbb{I}+\frac{i}{2}\mathbb{E}
\end{array}
\right] .
\end{eqnarray*}
\end{lemma}

\begin{proof}
This just says that we get $\mathbb{G}$ from shortening out the duplicate
set of labels. Indeed we have 
\begin{eqnarray*}
\underset{\overline{0}^{\prime }\cup \mathsf{k}^{\prime }\cup \underline{0}%
^{\prime }}{\mathrm{Schur}}\mathrm{\thinspace }\left[ 
\begin{array}{cc}
2\mathbb{I} & \sqrt{2}\mathbb{I} \\ 
\sqrt{2}\mathbb{I} & \mathbb{I}+\frac{i}{2}\mathbb{E}
\end{array}
\right] =2\mathbb{I}-2\left( \mathbb{I}+\frac{i}{2}\mathbb{E}\right) ^{-1}=%
\frac{i\mathbb{E}}{\mathbb{I}+\frac{i}{2}\mathbb{E}}
\end{eqnarray*}
which agrees with $\mathbb{G}=\frac{\mathbb{I}-\frac{i}{2}\mathbb{E}}{%
\mathbb{I}+\frac{i}{2}\mathbb{E}}\mathbb{-I}$ up to the sign.
\end{proof}

We are now able to state our main result which is the feedback reduction
rule in terms of the Stratonovich generator matrices.

\begin{theorem}
Let $\mathbf{E}$ be the Stratonovich generator matrix labeled by $0\cup 
\mathsf{e}\cup \mathsf{i}$ with $E_{\mathsf{ii}}$
invertible. The feedback
reduced Stratonovich generator (yielding (\ref{eq:SLH_feedback}) with $\eta=I_{\mathsf{i}}$) is 
\begin{eqnarray}
\mathbf{E}^{\mathrm{fb}}=\underset{\mathsf{i}}{\mathrm{Schur}}\mathrm{\thinspace }%
\mathbf{E}.
\label{eq:E_Schur}
\end{eqnarray}
\end{theorem}

That is, 
\begin{eqnarray*}
\left[ 
\begin{array}{cc}
E_{00}^{\mathrm{fb}} & E_{0\mathsf{e}}^{\mathrm{fb}} \\ 
E_{\mathsf{e}0}^{\mathrm{fb}} & E_{\mathsf{ee}}^{\mathrm{fb}}
\end{array}
\right] =\left[ 
\begin{array}{cc}
E_{00} & E_{0\mathsf{e}} \\ 
E_{\mathsf{e}0} & E_{\mathsf{ee}}
\end{array}
\right] -\left[ 
\begin{array}{c}
E_{0\mathsf{i}} \\ 
E_{\mathsf{ei}}
\end{array}
\right] E_{\mathsf{ii}}^{-1}\left[ 
\begin{array}{cc}
E_{\mathsf{i}0} & E_{\mathsf{ie}}
\end{array}
\right] .
\end{eqnarray*}

\begin{proof}  
Combining (\ref{eq:Schur_G}) with the lemma, we see that $\mathbb{G}^{\mathrm{%
fb}}$ can be written as successive Schur complements: 
\begin{eqnarray*}
\mathbb{G}^{\mathrm{fb}}=-\underset{\mathsf{i}}{\mathrm{Schur}}\mathrm{\thinspace }%
\underset{\overline{0}^{\prime }\cup \mathsf{k}^{\prime }\cup \underline{0}%
^{\prime }}{\mathrm{Schur}}\mathrm{\thinspace }\left[ 
\begin{array}{cc}
2\mathbb{I} & \sqrt{2}\mathbb{I} \\ 
\sqrt{2}\mathbb{I} & \mathbb{I}+\frac{i}{2}\mathbb{E}
\end{array}
\right] .
\end{eqnarray*}
We start with the doubled up set of labels $\left\{ \overline{0}\cup \mathsf{%
e}\cup \mathsf{i}\cup \underline{0}\right\} \cup \left\{ \overline{0}%
^{\prime }\cup \mathsf{e}^{\prime }\cup \mathsf{i}^{\prime }\cup \underline{0%
}^{\prime }\right\} $ and shortened by the duplicate labels $\left\{ 
\overline{0}^{\prime }\cup \mathsf{e}^{\prime }\cup \mathsf{i}^{\prime }\cup 
\underline{0}^{\prime }\right\} $ and then by the labels $\mathsf{i}$.
However, we could alternatively shortened first by $\mathsf{i}\cup \mathsf{i}%
^{\prime }$ and then by $\overline{0}^{\prime }\cup \mathsf{e}^{\prime }\cup 
\underline{0}^{\prime }$ to get the same result. Now 
\begin{eqnarray*}
&&\underset{\mathsf{i}\cup \mathsf{i}^{\prime }}{\mathrm{Schur}}\mathrm{%
\thinspace }\left[ 
\begin{array}{cc}
2\mathbb{I} & \sqrt{2}\mathbb{I} \\ 
\sqrt{2}\mathbb{I} & \mathbb{I}+\frac{i}{2}\mathbb{E}
\end{array}
\right] = \underset{\mathsf{i}\cup \mathsf{i}^{\prime }}{\mathrm{Schur}}\mathrm{%
\thinspace }\left[ 
\begin{array}{cccccccc}
2 & 0 & 0 & 0 & \sqrt{2} & 0 & 0 & 0 \\ 
0 & 2I_{\mathsf{e}} & 0 & 0 & 0 & \sqrt{2}I_{\mathsf{e}} & 0 & 0 \\ 
0 & 0 & 2I_{\mathsf{i}} & 0 & 0 & 0 & \sqrt{2}I_{\mathsf{i}} & 0 \\ 
0 & 0 & 0 & 2 & 0 & 0 & 0 & \sqrt{2} \\ 
\sqrt{2} & 0 & 0 & 0 & 1 & \frac{i}{2}\mathbb{E}_{0\mathsf{e}} & \frac{i}{2}%
\mathbb{E}_{0\mathsf{i}} & \frac{i}{2}\mathbb{E}_{00} \\ 
0 & \sqrt{2}I_{\mathsf{e}} & 0 & 0 & 0 & I_{\mathsf{e}}+\frac{i}{2}\mathbb{E}%
_{\mathsf{ee}} & \frac{i}{2}\mathbb{E}_{\mathsf{ei}} & \frac{i}{2}\mathbb{E}%
_{\mathsf{e}0} \\ 
0 & 0 & \sqrt{2}I_{\mathsf{i}} & 0 & 0 & \frac{i}{2}\mathbb{E}_{\mathsf{ie}}
& I_{\mathsf{i}}+\frac{i}{2}\mathbb{E}_{\mathsf{ii}} & \frac{i}{2}\mathbb{E}%
_{\mathsf{i}0} \\ 
0 & 0 & 0 & \sqrt{2} & 0 & 0 & 0 & 1
\end{array}
\right]
\end{eqnarray*}
\begin{eqnarray*}
=
\left[ 
\begin{array}{cccccc}
2 & 0 & 0 & \sqrt{2} & 0 & 0 \\ 
0 & 2I_{\mathsf{e}} & 0 & 0 & \sqrt{2}I_{\mathsf{e}} & 0 \\ 
0 & 0 & 2 & 0 & 0 & \sqrt{2} \\ 
\sqrt{2} & 0 & 0 & 1 & \frac{i}{2}\mathbb{E}_{0\mathsf{e}} & \frac{i}{2}%
\mathbb{E}_{00} \\ 
0 & \sqrt{2}I_{\mathsf{e}} & 0 & 0 & I_{\mathsf{e}}+\frac{i}{2}\mathbb{E}_{%
\mathsf{ee}} & \frac{i}{2}\mathbb{E}_{\mathsf{e}0} \\ 
0 & 0 & \sqrt{2} & 0 & 0 & 1
\end{array}
\right] 
-\left[ 
\begin{array}{cc}
0 & 0 \\ 
0 & 0 \\ 
0 & 0 \\ 
0 & \frac{i}{2}\mathbb{E}_{0\mathsf{i}} \\ 
0 & \frac{i}{2}\mathbb{E}_{\mathsf{ei}} \\ 
0 & 0
\end{array}
\right] \left[ 
\begin{array}{cc}
2I_{\mathsf{i}} & \sqrt{2}I_{\mathsf{i}} \\ 
\sqrt{2}I_{\mathsf{i}} & I_{\mathsf{i}}+\frac{i}{2}\mathbb{E}_{\mathsf{ii}}
\end{array}
\right] ^{-1}\left[ 
\begin{array}{cccccc}
0 & 0 & 0 & 0 & 0 & 0 \\ 
0 & 0 & 0 & 0 & \frac{i}{2}\mathbb{E}_{\mathsf{ie}} & \frac{i}{2}\mathbb{E}_{%
\mathsf{i}0}
\end{array}
\right]
\end{eqnarray*}
but 
\begin{eqnarray*}
\left[ 
\begin{array}{cc}
2I_{\mathsf{i}} & \sqrt{2}I_{\mathsf{i}} \\ 
\sqrt{2}I_{\mathsf{i}} & I_{\mathsf{i}}+\frac{i}{2}\mathbb{E}_{\mathsf{ii}}
\end{array}
\right] ^{-1}=\frac{1}{i\mathbb{E}_{\mathsf{ii}}}\left[ 
\begin{array}{cc}
I_{\mathsf{i}}+\frac{i}{2}\mathbb{E}_{\mathsf{ii}} & -\sqrt{2}I_{\mathsf{i}}
\\ 
-\sqrt{2}I_{\mathsf{i}} & 2I_{\mathsf{i}}
\end{array}
\right] 
\end{eqnarray*}
so that 
\begin{eqnarray*}
\underset{\mathsf{i}\cup \mathsf{i}^{\prime }}{\mathrm{Schur}}\mathrm{\thinspace 
}\left[ 
\begin{array}{cc}
2\mathbb{I} & \sqrt{2}\mathbb{I} \\ 
\sqrt{2}\mathbb{I} & \mathbb{I}+\frac{i}{2}\mathbb{E}
\end{array}
\right] =\left[ 
\begin{array}{cccccc}
2 & 0 & 0 & \sqrt{2} & 0 & 0 \\ 
0 & 2I_{\mathsf{e}} & 0 & 0 & \sqrt{2}I_{\mathsf{e}} & 0 \\ 
0 & 0 & 2 & 0 & 0 & \sqrt{2} \\ 
\sqrt{2} & 0 & 0 & 1 & \frac{i}{2}\mathbb{E}_{0\mathsf{e}}^{\mathrm{fb}} & 
\frac{i}{2}\mathbb{E}_{00}^{\mathrm{fb}} \\ 
0 & \sqrt{2}I_{\mathsf{e}} & 0 & 0 & I_{\mathsf{e}}+\frac{i}{2}\mathbb{E}_{%
\mathsf{ee}}^{\mathrm{fb}} & \frac{i}{2}\mathbb{E}_{\mathsf{e}0}^{\mathrm{fb}}
\\ 
0 & 0 & \sqrt{2} & 0 & 0 & 1
\end{array}
\right]
\end{eqnarray*}
where 
\begin{eqnarray*}
\mathbb{E}^{\mathrm{fb}}=\underset{\mathsf{i}}{\mathrm{Schur}}\mathrm{\thinspace }%
\mathbb{E}=\mathscr{H}\left( \underset{\mathsf{i}}{\mathrm{Schur}}\mathrm{%
\thinspace }\mathbf{E}\right) =\mathscr{H}\left( \mathbf{E}^{\mathrm{fb}%
}\right) .
\end{eqnarray*}
Shortening over the remaining duplicate labels $\left\{ \overline{0}^{\prime
}\cup \mathsf{e}^{\prime }\cup \underline{0}^{\prime }\right\} $ and
multiplying by minus, we recover $\mathbb{G}^{\mathrm{fb}}$. From the
expression obtained in the Lemma, we see that $\mathbf{E}^{\mathrm{fb}}=%
\underset{\mathsf{i}}{\mathrm{Schur}}$\thinspace $\mathbf{E}$ must be the
Stratonovich generating matrix.
\end{proof}

\bigskip

It is clear from the statement of the Theorem that there are situations where the 
feedback network is well-posed, i.e. $\mathscr{E}_{\mathsf{ii}}$ is invertible, but
the Schur complement is not invertible, i.e. $E_{\mathsf{ii}}$ is not invertible.
The following proposition gives a simple test of when the Schur complement exists.

\bigskip

\begin{proposition}
Let $S_{\mathsf{kk}}=\left[ 
\begin{array}{cc}
S_{\mathsf{ee}} & S_{\mathsf{ei}} \\ 
S_{\mathsf{ie}} & S_{\mathsf{ii}}
\end{array}
\right] $ be a scattering matrix decomposed with respect to a specification
of internal and external input/outputs $\mathsf{k}=\mathsf{e}\cup \mathsf{i}$%
, with $I_{\mathsf{k}}+S_{\mathsf{kk}}$ invertible. Let us set 
\begin{equation*}
\mathscr{S}_{\mathsf{ii}}\triangleq S_{\mathsf{ii}}-S_{\mathsf{ie}}\left( I_{\mathsf{i}%
}+S_{\mathsf{ii}}\right) ^{-1}S_{\mathsf{ie}}.
\end{equation*}
Then the Stratonovich matrix block $E_{\mathsf{ii}}$ is invertible if and
only if $I_{\mathsf{i}}-\mathscr{S}_{\mathsf{ii}}$ is invertible.
\end{proposition}

\begin{proof}
Starting from the identity $E_{\mathsf{kk}}=\frac{2}{i}\left( I_{\mathsf{k}%
}+S_{\mathsf{kk}}\right) ^{-1}\left( I_{\mathsf{k}}-S_{\mathsf{kk}}\right) $
we obtain the sub-block
\begin{equation*}
E_{\mathsf{ii}}=\frac{2}{i}\left( I_{\mathsf{i}}+\mathscr{S}_{\mathsf{ii}}\right)
^{-1}\left( I_{\mathsf{i}}- \mathscr{S}_{\mathsf{ii}}\right) 
\end{equation*}
using the same methods as used in the previous proposition. This formal
expression shows precisely when the block $E_{\mathsf{ii}}$ is invertible,
and we arrive at the desired conclusion.
\end{proof}

\subsection{Beam-splitter Example}
As an example we consider the simple example of a beam-splitter with
Stratonovich matrix
\begin{equation*}
E_{\mathsf{kk}}=\left[ 
\begin{array}{cc}
E_{\mathsf{ee}} & E_{\mathsf{ei}} \\ 
E_{\mathsf{ie}} & E_{\mathsf{ii}}
\end{array}
\right] \equiv \left[ 
\begin{array}{cc}
\alpha  & \beta  \\ 
\beta ^{\ast } & \gamma 
\end{array}
\right] 
\end{equation*}
with $\alpha ,\beta $ real and $\beta $ complex. A simple algebra shows that
\begin{equation*}
S_{\mathsf{kk}}\equiv \frac{1}{1+\frac{i}{2}\left( \alpha +\gamma \right) -%
\frac{1}{4}\left( \alpha \gamma -|\beta |^{2}\right) }\left[ 
\begin{array}{cc}
1+\frac{i}{2}\left( \gamma -\alpha \right) +\frac{1}{4}\left( \alpha \gamma
-|\beta |^{2}\right)  & -i\beta  \\ 
-i\beta ^{\ast } & 1+\frac{i}{2}\left( \alpha -\gamma \right) +\frac{1}{4}%
\left( \alpha \gamma -|\beta |^{2}\right) 
\end{array}
\right] 
\end{equation*}
which gives the general form of a beam-splitter derivable from a Stratonovich
form. 

\begin{figure}[htbp]
	\centering
		\includegraphics[width=0.50\textwidth]{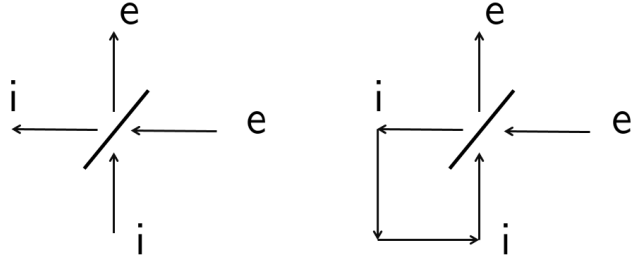}
	\caption{(Left) a beam-splitter with matrix $S_{\mathsf{kk}}$; (Right) its feedback reduction.}
	\label{fig:beamsplitter_loop_example}
\end{figure}

We see that if $\gamma =0\left( =E_{\mathsf{ii}}\right) $ then the Schur
complement is not defined so the Theorem does not apply automatically.
Nevertheless we see that $\mathscr{E}_{\mathsf{ii}}=\gamma -\frac{i}{2}\frac{%
|\beta |^{2}}{1+\frac{i}{2}\alpha }$ and so the network may be well-posed  even if 
$\gamma =0$, and it is instructive to look at this case. Here we have
\begin{equation*}
S_{\mathsf{kk}}\equiv \frac{1}{1+\frac{i}{2}\alpha +\frac{1}{4}|\beta |^{2}}%
\left[ 
\begin{array}{cc}
1-\frac{i}{2}\alpha -\frac{1}{4}|\beta |^{2} & -i\beta  \\ 
-i\beta ^{\ast } & 1+\frac{i}{2}\alpha -\frac{1}{4}|\beta |^{2}
\end{array}
\right] 
\end{equation*}
prior to feedback, and for all parametrizations we have that the feedback
reduced scattering is
\begin{equation*}
S_{\mathsf{ee}}^{\text{fb}}\equiv -1.
\end{equation*}
The answer is of course not something which has a well-defined Stratonovich
form: in fact we have only $\lim_{\varepsilon \rightarrow \infty }\frac{1-%
\frac{i}{2}\varepsilon }{1+\frac{i}{2}\varepsilon }=-1$ but that this can
not be realized for finite $\varepsilon $.

\subsection{The Series Product}

Let us recall that the series product can be written in terms of Belavkin-Holevo matrices as
$\mathbb{V}_{\mathrm{series}}=\mathbb{V}_{2}\mathbb{V}_{1}$. If both $\mathbb{V}_{2}$ 
and $\mathbb{V}_{1}$ come from Stratonovich matrices $\mathbb{E}_{2}$ and $\mathbb{E}_{1}$ respectively, then the
corresponding $\mathbb{E}_{\mathrm{series}}$ should be expressible in terms of $\mathbb{E}_{2}$ and $\mathbb{E}_{1}$.

\begin{proposition}
Suppose that both $\mathbb{V}_{2}$ and $\mathbb{V}_{1}$ come from Stratonovich matrices 
$\mathbb{E}_{2}$ and $\mathbb{E}_{1}$ respectively,  then their series product 
$\mathbb{V}_{\mathrm{series}}=\mathbb{V}_{2}\mathbb{V}_{1}$
has the Stratonovich matrix
\begin{equation*}
\mathbb{E}_{\text{series}}=\left( \mathbb{I}+\frac{i}{2}\mathbb{E}%
_{2}\right) ^{-1}\left( \mathbb{E}_{1}+\mathbb{E}_{2}\right) \left( \mathbb{I%
}-\frac{1}{4}\mathbb{E}_{2}\mathbb{E}_{1}\right) ^{-1}\left( \mathbb{I}+%
\frac{i}{2}\mathbb{E}_{2}\right) .
\end{equation*}
\bigskip 
\end{proposition}

\begin{proof}
Let us begin by noting that the Cayley transformation $F\mapsto V\left(
F\right) =\frac{I-F}{I+F}$ is a map from matrices $\left\{ F:I+F\text{ is
invertible}\right\} $ to $\left\{ V:I+V\text{ is invertible}\right\} $.
Indeed, if $V=\frac{I-F}{I+F}$ then $F=\frac{I-V}{I+V}$, so $V^{-1}\left(
\cdot \right) \equiv V\left( \cdot \right) $. We wish to solve $V\left(
F_{3}\right) =V\left( F_{2}\right) V\left( F_{1}\right) $ for $F_{3}$ given $%
F_{1}$ and $F_{2}$. The identity $V\left( F_{3}\right) =V_{2}V\left(
F_{1}\right) $ can be rearranged to give
\begin{equation*}
F_{3}=\left[ \left( I-V_{2}\right) +\left( I+V_{2}\right) F_{1}\right] \left[
\left( I+V_{2}\right) +\left( I-V_{2}\right) F_{1}\right] ^{-1}
\end{equation*}
and setting $F_{2}=V^{-1}\left( V_{2}\right) \equiv \frac{I-V_{2}}{I+V_{2}}$
leads to
\begin{equation*}
F_{3}=\left( I+V_{2}\right) \left( F_{1}+F_{2}\right) \left(
I+F_{2}F_{1}\right) ^{-1}\left( I+V_{2}\right) ^{-1}
\end{equation*}
and noting that $I+V_{2}=2\frac{I}{I+F_{2}}$ gives the result.
\end{proof}

\subsection{Adjacency Matrices}
The adjacency operator corresponding to a given permutation $\sigma $ on a
set of $n$ labels will be denoted by $\eta \left( \sigma \right) $. That is,
\begin{equation*}
\left[ \eta \left( \sigma \right) \right] _{jk}=\left\{ 
\begin{array}{cc}
1, & j=\sigma \left( k\right) ; \\ 
0, & \text{otherwise.}
\end{array}
\right. 
\end{equation*}
A natural question is when does there exist a $E_{\mathsf{kk}}$ such that $%
\eta \left( \sigma \right) =\frac{I_{\mathsf{k}}-\frac{i}{2}E_{\mathsf{kk}}}{%
I_{\mathsf{k}}+\frac{i}{2}E_{\mathsf{kk}}}$. For this to be possible we need
that $I_{\mathsf{k}}-\eta \left( \sigma \right) $ is invertible, so that $E_{\mathsf{kk}}
\equiv \frac{2}{i} \frac{I_n-\eta (\sigma )}{I_n +\eta (\sigma )}$.

This mapping $:\sigma \mapsto \eta \left( \sigma \right) $ from the set of
permutations on $n$ labels to the $n\times n$ matrices is a reducible
representation of the permutation group: indeed each $\eta \left( \sigma
\right) $ has eigenvalue unity for the eigenvector $\left[ 1,\cdots ,1\right]
^{\top }$ and this is a nontrivial invariant subspace. More generally, we
have that the spectrum (including explicit degeneracies) of $\eta \left( \sigma \right) $ is the multi-set $%
\cup _{k\geq 1}\left\{ k\text{-th roots of unity}\right\} ^{n_{k}\left(
\sigma \right) }$ where $n_{k}\left( \sigma \right) $ counts the number of
cycles of length $k$ in the permutation. In particular, we note that $-1$ is
not in the spectrum if and only if there are no even cycles in $\sigma $
\cite{Diaconis}. 

As a result, we have that $I_{n}+ \eta \left( \sigma \right) $ does not have zero as an eigenvalue
if and only if $\sigma $ has no even cycles. This is the condition for $\eta (\sigma )$ to be expressed as a
Cayley transorm of some $2\times 2$ matrix $E_{\mathsf{kk}}$.
For instance, the swap gate
\begin{equation*}
\eta =\left[ 
\begin{array}{cc}
0 & 1 \\ 
1 & 0
\end{array}
\right] 
\end{equation*}
corresponding to the cycle $\left( 12\right) $ is the simplest adjacency matrix that cannot be expressed in this way. 

\section{Conclusions}
The Stratonovich form of the quantum stochastic calculus has the advantage of revealing a Hamiltonian structure. This is readily seen in the rule for concatenating components to get $\Upsilon (t)$ in (\ref{eq:S_parallel}) - and so (\ref{eq:E_parallel}). The fact that the feedback reduction rule has the direct form of a Schur complement, that is equation (\ref{eq:E_fb}), is however an unexpected feature. We might mention
that the Schur complement has previously emerged as the appropriate tool in adiabatic elimination results in the SLH formalism
\cite{BSvH,GNW,GN}.

In the paper, we have restricted our attention to just the proof of the mathematical form of the feedback reduced Stratonovich matrix.
This is done at the very broad level of generality offered by the SLH formalism, and have deferred the application to specific models
for a latter publication. It should be mentioned that situations such as completely closed feedback loops served as a motivating problem for this publication as the Schur complement form already appears in a restricted sense, and strongly suggestive of the general result presented here in Theorem 1. We have also restricted to the vacuum case for the input fields. The issue of introducing time delays into the
feedback connections has also been ignored - we assume the validity of an instantaneous feedback limit - however we note that there has
been some promising developments in this direction for linear quantum systems \cite{Tabak}.

While the Stratonovich form of the feedback reduction rule is much simpler that the It\={o} form (\ref{eq:SLH_feedback}), it is not
true that it is universally easier to work with. For instance, the series product formula is a corollary to the concatenation rule 
(\ref{eq:SLH_parallel}) and the feedback rule (\ref{eq:SLH_feedback}), and so must be derivable from the two equivalent formulations in terms of Stratonovich matrices of coefficients. However this derivation is very involved. It is therefore the case that certain operations, such as putting systems in series are better handled with the It\={o}, or SLH, form while other operations, such as feedback reduction may be better handled in the Stratonovich form. The situation is not unlike classical circuit theory where one chooses judicially the form of the immittances: namely impedances for components is series in a network, and admittances (their inverses) for components placed in parallel in a network. While analogy is not exact, it does suggest that a hybrid use of the It\={o} and Stratonovich rules may be very useful for calculating the SLH characteristics of complex quantum feedback networks.

\section{Acknowledgements}
The paper was motivated by discussions at the workshop Quantum and Nano-Control at the Institute of Mathematics and Its Applications in Minnesota during April 2016, and the kind support and hospitality of the IMA is greatly acknowledged. In particular, the author wishes to thank Professors Ian Petersen, Matthew James and Hideo Mabuchi for valuable conversations and comments.


\begin{thebibliography}{99}
\bibitem{HP}  R. L. Hudson and K. R. Parthasarathy,
\textit{Quantum Ito's formula and stochastic evolutions,} 
Commun. Math. Phys. \textbf{93}, 301-323
(1984)

\bibitem{partha}  K.R. Parthasarathy,  An Introduction to Quantum Stochastic
Calculus, Birkhauser, Berlin (1992)

\bibitem{Gardiner}  C. Gardiner and P. Zoller, Quantum Noise: A Handbook of
Markovian and Non-Markovian Quantum Stochastic Methods with Applications to
Quantum Optics, 2nd ed., ser. Springer Series in Synergetics. Springer, (2000)

\bibitem{GJ_QFN}  J.E. Gough, M.R. James, 
\textit{Quantum feedback networks: Hamiltonian formulation}, 
Commun. Math. Phys. \textbf{287}, 1109 (2009)

\bibitem{GJ_Series}  J.E. Gough, M.R. James, 
\textit{The Series Product and Its Application to Quantum Feedforward and Feedback Networks}, 
IEEE Trans. on Automatic Control \textbf{54}, 2530 (2009)


\bibitem{I_12}  S. Iida, M. Yukawa, H. Yonezawa, N. Yamamoto, A. Furusawa,
\textit{Experimental demonstration of coherent feedback control on optical field squeezing},
IEEE Trans. Auto. Control, \textbf{57} , 8, 2045 - 2050 (2012)

\bibitem{KNPM_10}  J. Kerckhoff, H.I. Nurdin, D.S. Pavlichin, H. Mabuchi,
\textit{Designing quantum memories with embedded control: photonic circuits for autonomous quantum error correction},
Phys. Rev. Lett. 105, 040502 (2010)

\bibitem{Crisafulli_2010}
O. Crisafulli, N. Tezak, D.B.S. Soh, M.A. Armen  \& H. Mabuchi, 
\textit{Squeezed light in an optical parametric oscillator network with coherent feedback quantum control}, 
Opt. Express 21, 18371 (2013).

\bibitem{SPBTHM_14}  C. Santori, J.S. Pelc, R.G. Beausoleil, N. Tezak, R.
Hamerly, H. Mabuchi, 
\textit{Quantum Noise in Large-Scale Coherent Nonlinear Photonic Circuits}, Phys. Rev. Applied 1, 054005 (2014).

\bibitem{ZJLYXP_2015} Y. Zhou, X. Jia, F. Li, J. Yu, C. Xie \& K. Peng,
\textit{Quantum Coherent Feedback Control for Generation System of Optical Entangled State,}
Scientific Reports, 5:11132  (2015)


\bibitem{Chebotarev}  A.M. Chebotarev, 
\textit{The quantum stochastic equation is unitarily equivalent to a symmetric boundary value problem for the Schrödinger equation},
Math. Notes 61, No. 4, 510 (1997)

\bibitem{G_Wong-Zakai}  J.E. Gough, 
\textit{Quantum Stratonovich calculus and the quantum Wong-Zakai theorem}, 
J. Math. Phys., vol. 47, no. 113509, (2006)

\bibitem{G_JMP_2015}  J.E. Gough, 
\textit{Characteristic operator functions for quantum input-plant-output models and coherent control}, 
Journ. Math. Physics 56, 013506 (2015)

\bibitem{G_OSID_2015}  J.E. Gough, 
\textit{The Global versus Local Hamiltonian Description of Quantum Input-Output Theory}, 
Open Systems \& Information Dynamics Vol. 22, No. 2, 1550009 (2015) 


\bibitem{QHDL} N. Tezak, A. Niederberger, D.S. Pavlichin, G. Sarma, and H. Mabuchi. 
\textit{Specification of photonic circuits using quantum hardware description language},
Philosophical Transactions of the Royal
Society of London A: Mathematical, Physical and Engineering Sciences, 370(1979):5270–5290, (2012)

\bibitem{Belavkin_1988} V.P. Belavkin, 
\textit{A new form and $\star$-algebraic structure of quantum stochastic integrals in Fock space},
Rend. Sem. Mat. Fis. Milano, 58, 177-193, (1988)

\bibitem{Holevo_1989} A.S. Holevo, 
\textit{Stochastic representation of quantum dynamical semi- groups},
Trudy Mat. Inst. Steklov, 191, 130-139 (1989)

\bibitem{ST} O.G. Smolyanov, A. Truman, 
\textit{The Gough-James Theory of Quantum Feedback Networks in Belavkin's Representation}, 
Doklady Math., Vol. \textbf{82}, No. 3, 974-977 (2010)


\bibitem{Schur} F. Zhang (Ed.) The Schur Complement and its Applications, Springer Series Numerical Methods and Algorithms, Vol. 4, Springer,
New York, (2005).

\bibitem{Diaconis} P. Diaconis, Group Representations in Probability and Statistics, IMS
Lecture Notes - Monograph Series Volume 11, Institute of Mathematical Statistics, Hayward, (1988)

\bibitem{BSvH}
L. Bouten, R. van Handel, A. Silberfarb, 
\textit{Approximation and limit theorems for quantum stochastic models with unbounded coefficients}, J. Funct. Analysis. 254, 3123-3147 (2008)

\bibitem{GNW}
J.E. Gough, H.I. Nurdin, S. Wildfeuer,
\textit{Commutativity of the adiabatic elimination limit of fast oscillatory components and the instantaneous feedback limit in quantum feedback networks},
Journal of Mathematical Physics, 51(12), pp. 123518-1--123518-25 (2010)

\bibitem{GN}
H.I. Nurdin, J.E. Gough,
\textit{On structure-preserving transformations of the Ito generator matrix for model reduction of quantum feedback networks},
Phil. Trans. R. Soc. A 28, vol. 370 no. 1979 5422-5436, November (2012)

\bibitem{Tabak}
G. Tabak and H. Mabuchi, 
\textit{Trapped modes in linear quantum stochastic networks with delays},
EPJ Quantum Technology 20163:3 (2016)

\end{thebibliography}
\end{document}